 \documentclass[conference, onecolumn,11pt]{IEEEtran}
\IEEEoverridecommandlockouts

\usepackage{graphicx}
\usepackage{cite}
\usepackage{color}
\usepackage{amsfonts}
\usepackage{amsthm}
\usepackage{amsmath}
\usepackage{amssymb}
\usepackage{subfig}
\usepackage{fixltx2e}

\usepackage{pseudocode}
\usepackage[]{fullpage}
\graphicspath{{figs/}}
\newtheorem{theorem}{Theorem}
\newtheorem{lemma}{Lemma}
\newtheorem{remark}{Remark}
\newtheorem{definition}{Definition}

\newtheorem{example}{Example}

\newcommand{\prot}{\mathbf{P}}
\usepackage{mathrsfs}
\begin{document}

\title{{Coded Cooperative Data Exchange\\ for a Secret Key}
}

\author{
\IEEEauthorblockN{Thomas~A.~Courtade}
\IEEEauthorblockA{Dept. of Electrical Engineering and Computer Sciences\\University of California, Berkeley\\
              Email: courtade@eecs.berkeley.edu\vspace{-0ex}}
\and
\authorblockN{Thomas R. Halford}
\authorblockA{TrellisWare Technologies\\ San Diego, CA\\
Email: thalford@trellisware.com\vspace{-0ex}}
\thanks{This work is supported by in part by the NSF Center for Science of Information under grant agreement CCF-0939370.}
\thanks{This paper will appear in part at the 2014 International Symposium on Information Theory \cite{bib:PK_CCDE_ISIT2014}.}
}

\maketitle

\begin{abstract}
We consider a  coded cooperative data exchange problem with the goal of generating a secret key.  Specifically, we investigate the number of public transmissions required for a set of clients to agree on a secret key with probability one,  subject to the constraint  that it remains private from an eavesdropper.

Although the problems are closely related, we prove  that secret key generation with fewest number of linear transmissions is \textsf{NP}-hard, while it is known that the analogous problem in  traditional cooperative data exchange  can be solved in polynomial time. In doing this, we completely characterize the best-possible performance of linear coding schemes, and also prove that linear codes can be strictly suboptimal.  Finally, we extend the single-key results to characterize the minimum number of public transmissions required to generate a desired integer number of statistically independent secret keys.
\end{abstract}

\section{Introduction}\label{sec:intro}
In this paper, we consider a  cooperative data exchange problem with the goal of generating a secret key.  Primarily, we study the number of public transmissions required for a set of clients to agree on a secret key,  subject to the constraint  that it remains private from an eavesdropper.

In an asymptotic setting, the reciprocal relationship between secret key (SK) capacity and communication for omniscience was revealed in the pioneering work  \cite{bib:CsiszarNarayanIT2004} by Csisz\'{a}r and Narayan.  They showed that the maximum rate at which secrecy can be generated by a collection of terminals is in one-to-one correspondence with the minimum  rate required for those same terminals to communicate for omniscience.  Though they characterized the minimum communication rate required to attain omniscience, Csisz\'ar and Narayan left characterizing the minimum communication rate required to generate a maximum-rate SK  as an open problem \cite[Section VI]{bib:CsiszarNarayanIT2004}.  

In \cite{bib:ElRouayheb2010ITW}, El Rouayheb \emph{et al}.\ introduced a non-asymptotic, combinatorial version of Csisz\'{a}r and Narayan's communication for omniscience problem, which they called \emph{(coded) cooperative data exchange} (CCDE). Since its introduction, this problem has received significant attention from many researchers (see \cite{bib:CourtadeWesel2013} and the references therein). Algorithms and heuristics for solving the CCDE problem were presented in \cite{bib:SprintsonSadeghiISIT2010,bib:SprintsonQshine2010,bib:CourtadeWeselAllerton2010} for broadcast networks and in \cite{bib:CourtadeWeselAllerton2011,bib:GonenLangberg12} for multihop networks. Moreover, a number of authors have considered generalizations of the CCDE problem to model various practical system considerations \cite{bib:Milosavljevic2012,bib:Milosavljevic,bib:TajbakhshSadeghiShams2011,bib:OzgulSprintsonITA2011,bib:HouHsuSprintson13}.

In \cite{bib:CourtadeWeselAllerton2011, bib:CourtadeWesel2013}, Courtade and Wesel showed that the omniscience-secrecy relations in  \cite{bib:CsiszarNarayanIT2004} translate nicely to the combinatorial CCDE setting.  
Building on \cite{bib:CourtadeWeselAllerton2011,bib:CourtadeWesel2013}, we presently investigate the number of public transmissions required for a set of clients to agree on a SK with probability one, subject to the constraint  that it remains private from an eavesdropper.  In doing so, we address a combinatorial analog of  Csisz\'ar's open problem. On this note, we remark that independent of the work on CCDE, Chan considered a closely related finite linear source model and  gave suboptimal bounds on the public transmission block length required for perfect SK agreement \cite{chan2011linear}. Thus, our results give a definitive solution to Chan's problem under the CCDE model. 

Related to the present work is the minimum communication rate required to generate a maximum-rate SK in the asymptotic setting. In \cite{tyagi2013}, Tyagi gave a multi-letter expression characterization for this rate in the two-terminal case in terms of the \emph{r-rounds interactive common information}. Recently, Mukherjee and Kashyap considered extensions to the multi-terminal case \cite{bib:MukherjeeKashyap2014}. %

Despite the similarity in spirit, the asymptotic setting of \cite{tyagi2013}  and the combinatorial setting of the present paper are considerably different in nature, and the proof techniques used are orthogonal. That said, all of these results shed light on the fundamentally different natures of SK generation at minimum communication rate and SK generation via communication for omniscience.

The weakly secure CCDE problem introduced in \cite{yan2013algorithms} is also related to our work. The goal of the weakly secure CCDE problem is to communicate for omniscience while revealing as little information as possible to an eavesdropper. This is closely related to the CCDE under a privacy constraint problem studied in \cite{bib:CourtadeWesel2013}. Yan and Sprintson designed coding schemes that solve the weakly secure CCDE problem while revealing as little as information as possible to an eavesdropper \cite{yan2013algorithms}. Improvements to these schemes can be derived using the codes described in \cite{dau2014existence}. The primary distinction between the present setting and that of the weakly secure CCDE problem is that we only aim to generate a SK; we do not require that the nodes communicate for omniscience nor do we require that the SK corresponds to any given message.   

Finally, we remark that in a recent paper \cite{bib:HalfordCourtadeChugg2013}, Halford \emph{et al.}\ developed practical protocols for SK generation in ad hoc networks based on the CCDE problem. Briefly, a scenario was studied wherein the protocol designer controls the initial distribution of master keys so that secret keys can later be efficiently generated among arbitrary groups of clients. The results given in the present paper establish limits and suggest design rules for such protocols.

\subsection*{Our Contributions}

Given the close connection between CCDE and SK generation, we show two surprising results.  First, we prove that finding an optimal (linear) coding scheme for SK generation  is \textsf{NP}-hard, while it is known that the analogous CCDE problem is in \textsf{P}.   In doing this, we completely characterize the attainable performance for linear coding schemes in terms of hypergraph connectivity.  Second, despite linear codes being optimal for CCDE, we demonstrate that they can be strictly suboptimal for SK generation. Several ancillary results are also proved.%

This paper is organized as follows. Section \ref{sec:model}  formally defines our system model and reviews relevant results on CCDE.  In Section \ref{sec:Results}, we state and prove our main results for the generation of a single SK.  Section \ref{sec:multipleSK} characterizes the minimum number of public transmissions required to generate multiple SKs, and  Section \ref{sec:conclusion} delivers concluding remarks.

\section{System Model and Preliminaries}\label{sec:model}
We first establish  basic notation.
Throughout, we use calligraphic notation to denote sets. For two sets $\mathcal{A}\subset \mathcal{B}$, we write $\mathcal{B}\backslash \mathcal{A}$ to denote those elements in $\mathcal{B}$, but not in $\mathcal{A}$.  If $\mathcal{A}$ is a singleton set (i.e., $\mathcal{A}=\{a\}$), then we often use the notation $\mathcal{B}-a \triangleq \mathcal{B}\backslash \{a\}$ for convenience.  %
We define $\mathbb{Z}$ to be the set of integers.  For  positive $m\in \mathbb{Z}$, we use the shorthand notation $[m] \triangleq  \{1,2,\dots,m\}$.
Finally, for random variables $X,Y$,  we write $I(X;Y)$ for  the  mutual information between $X$ and $Y$.

\subsection{System Model}

Throughout, we consider networks defined by a set of $n$ clients (i.e., terminals) $\mathcal{C}=\{c_1,c_2,\dots,c_n\}$, a positive integer $m$, and a family of finite sets $\{\mathcal{I}_1, \mathcal{I}_2, \dots, \mathcal{I}_n\}$ (each $\mathcal{I}_j\subseteq [m]$ and  $\cup_{j=1}^n \mathcal{I}_j =[m]$) in the following way.  Define the random (column) vector $\underbar{X}\triangleq \left[ X_1,X_2,\dots, X_m\right]^T$, where each  $X_i$  is a discrete random variable with equiprobable distribution on a finite field $\mathbb{F}$, and $(X_1,X_2,\dots,X_m)$ are mutually independent\footnote{For technical reasons, we assume $|\mathbb{F}|>n$.}.  The random variables $\{X_i\}_{i=1}^m$ are called \emph{messages}, and $\{X_i : i\in \mathcal{I}_j\}$ is the set of messages initially held by client $c_j\in \mathcal{C}$.  In other words, $\mathcal{I}_j$ defines the indices of messages initially held by client $c_j$, for $j=1,\dots,n$. Throughout, $n$ will always denote the number of clients; since the sets $\mathcal{I}_j$ are always indexed by $j\in[n]$, we will use the shorthand notation $\{\mathcal{I}_j\}$ to denote the family $\{\mathcal{I}_1,\mathcal{I}_2, \dots,\mathcal{I}_n\}$.

We adopt the communication model which is standard in index coding and CCDE problems.  That is, we  consider transmission schemes consisting of a finite number of communication rounds.  In each round, a single client broadcasts an element of $\mathbb{F}$ (which can be a function of the messages initially held by that client and all previous transmissions) to all other clients over an error-free channel.  It is further assumed that all clients have knowledge of the index sets $\mathcal{I}_1, \dots, \mathcal{I}_n$, and thus follow a protocol which is mutually agreed upon.  We will elaborate on the definition of a transmission protocol in the next subsection.

 \subsection{Transmission Protocols Defined}
 
 For a network defined by $\{\mathcal{I}_j\}$, a transmission protocol $\prot$ (or simply, a protocol $\prot$) consisting of $t$ communication rounds is defined by $n$ encoding functions $\{\mathsf{f}_1,  \mathsf{f}_2, \dots,  \mathsf{f}_n\}$, and a $t$-tuple $(i_1,i_2,\dots,i_t)$, where $i_k\in [n]$ indicates which client transmits during communication round $k$.  More specifically, during communication round $k$, client $c_{i_k}$ transmits %
 \begin{align}
 \mathsf{f}_{i_k}\!\Big(\{X_j : j\in \mathcal{I}_{i_k} \},k, \{\mathsf{f}_{i_{\ell}}\}_{\ell=1}^{k-1} \Big) \in \mathbb{F}, \label{encFn}
 \end{align}
 where we have abbreviated the transmitted symbols in rounds $\ell \in[k-1]$ by $\{\mathsf{f}_{i_{\ell}}\}_{\ell=1}^{k-1}$.  For a given transmission protocol $\prot$ requiring $t$ communication rounds, we let $\mathbf{T}(\underbar{X},\prot) \in \mathbb{F}^{t}$  be the column vector with $k^{\mathsf{th}}$ entry  equal to $\mathsf{f}_{i_k}\!\left(\{X_j : j\in \mathcal{I}_{i_k} \}, k,\{\mathsf{f}_{i_{\ell}}\}_{\ell=1}^{k-1} \right)$.  Letting $\|\cdot \|$ be the length function, we have $\|\mathbf{T}(\underbar{X},\prot ) \| = t$.  Note that $\mathbf{T}(\underbar{X},\prot)$ is a random variable since it is a function of the random vector $\underbar{X}$.  Generally, the transmission protocol under consideration will be clear from context.  Hence, we abbreviate $\mathbf{T}(\underbar{X})\triangleq \mathbf{T}(\underbar{X},\prot)$ for convenience when there is no ambiguity.

A transmission protocol is said to be \emph{linear} (over $\mathbb{F}$) if the encoding functions  $\{\mathsf{f}_1,  \mathsf{f}_2, \dots, \mathsf{f}_n\}$ are of the form 
\begin{align}
 \mathsf{f}_{i_k}\!\left(\{X_j : j\in \mathcal{I}_{i_k} \},k, \{\mathsf{f}_{i_{\ell}}\}_{\ell=1}^{k-1}\right) = \sum_{ j} \alpha_j^{(k)} X_j, %
\end{align}
 where $\alpha_j^{(k)}\in \mathbb{F}$ can be interpreted as the encoding coefficient for message $j$ during communication round $k$.  In this case, we can express 
$\mathbf{T}(\underbar{X}) = A\underbar{X}$,
 where $A\in \mathbb{F}^{t \times m}$ assuming  the definitions $t \triangleq \|\mathbf{T}(\underbar{X}) \|$ and $m \triangleq |\cup_j \mathcal{I}_j|$.  Hence, the \emph{encoding matrix} $A$ provides a succinct description of a linear transmission protocol.  %
 Note that the order of transmissions corresponding to a linear protocol is inconsequential.
 
\subsection{Transmission Protocols for Omniscience} 
A transmission protocol $\prot$ is said to achieve \emph{omniscience} if there exist decoding functions $\{\mathsf{g}_1,  \mathsf{g}_2, \dots,  \mathsf{g}_n\}$ which satisfy
 \begin{align}
  \mathsf{g}_{j}\!\left(\{X_i : i\in \mathcal{I}_{j} \},   \mathbf{T}(\underbar{X},\prot)  \right) = \underbar{X} \mbox{~~for each $j\in[n]$}
 \end{align}
 with probability 1.

Before proceeding,  let $\mathsf{M}^{\star}\!\left(\{\mathcal{I}_j\}\right)$ denote the optimal value of the following integer linear program (ILP):
\begin{align}
\mathsf{minimize:}&~~ \sum_{j\in[n]} a_j \label{ILP}\\
\mathsf{subject~to:}&~~ \sum_{j\in\mathcal{S}} a_j  \geq \left|\bigcap_{j\in \bar{\mathcal{S}} } \bar{\mathcal{I}_j} \right| \mbox{~~for all nonempty $\mathcal{S}\subset [n]$}\notag\\
&~~ a_j \in \mathbb{Z} \mbox{~~for all $j\in [n]$,}\notag
\end{align}
where $\bar{\mathcal{I}_i}\triangleq \left( \cup_j \mathcal{I}_j \right)\backslash \mathcal{I}_i$ and $\bar{\mathcal{S}}\triangleq [n]\backslash \mathcal{S}$.
The quantity $\mathsf{M}^{\star}\!\left(\{\mathcal{I}_j\}\right)$ will play an important role in our treatment due to its inherent connection to the communication for omniscience, which is made explicit by the following theorem\footnote{Theorem \ref{thm:UR} essentially appeared in the given form in \cite{bib:CourtadeWeselAllerton2011}.  However, it was independently discovered by Milosavljevic et al. \cite{bib:Milosavljevic} and Chan \cite{chan2011linear} at roughly the same time.}. 
 \begin{theorem}\label{thm:UR} {\cite[Theorem 2]{bib:CourtadeWesel2013}} If a protocol $\prot$ achieves omniscience, then $\|\mathbf{T}(\underbar{X},\prot)\| \geq \mathsf{M}^{\star}\!\left(\{\mathcal{I}_j\}\right)$.  Conversely, there always exists a linear protocol $\prot_{\mathrm{L}}$ that achieves omniscience and has $\|\mathbf{T}(\underbar{X},\prot_{\mathrm{L}})\| = \mathsf{M}^{\star}\!\left(\{\mathcal{I}_j\}\right)$. \end{theorem}
 
 Theorem \ref{thm:UR} addresses the central issue in the CCDE problem, which primarily investigates the number of transmissions required to achieve omniscience.  We remark that this is not equivalent to characterizing the minimum communication \emph{rate} required for omniscience (as would be the case in the original communication for omniscience problem \cite{bib:CsiszarNarayanIT2004}) due in part to the integrality constraint on the number of transmissions.

\subsection{Transmission Protocols for Secret Keys}
  
A transmission protocol (with corresponding transmission sequence $\mathbf{T}(\underbar{X})$) generates a secret key (SK) if there exist decoding functions $\{\mathsf{k}_1,\mathsf{k}_2,\dots, \mathsf{k}_n\}$ which satisfy the following three properties:
\begin{enumerate}
\item[(i)] For all $j\in [n]$, and with probability 1, 
$$
\!\!\!\!\!\!\!\!\mathsf{k}_{j}\!\left(\{X_i   : i\in \mathcal{I}_{j} \},  \mathbf{T}(\underbar{X})  \right) \notag=\mathsf{k}_{1}\!\left(\{X_i  : i\in \mathcal{I}_{1} \},  \mathbf{T}(\underbar{X})   \right).
$$
\item[(ii)] $\mathsf{k}_{1}\!\left(\{X_i  : i\in \mathcal{I}_{1} \},  \mathbf{T}(\underbar{X})   \right)$ is equiprobable on $\mathbb{F}$.
\item[(iii)] $I\left(\mathsf{k}_{1}\!\left(\{X_i  : i\in \mathcal{I}_{1} \},  \mathbf{T}(\underbar{X})   \right);  \mathbf{T}(\underbar{X})  \right)=0$.
\end{enumerate} 

In words, requirement (iii) guarantees that the public transmissions $ \mathbf{T}(\underbar{X})$ reveal no information about  $\mathsf{k}_{1}\!\left(\{X_i  : i\in \mathcal{I}_{1} \},  \mathbf{T}(\underbar{X})  \right)$. 
Requirement (i) asserts that all clients $c_j \in \mathcal{C}$ can compute $\mathsf{k}_{1}\!\left(\{X_i  : i\in \mathcal{I}_{1} \},  \mathbf{T}(\underbar{X})  \right)$.  For these reasons, $\mathsf{k}_{1}\!\left(\{X_i  : i\in \mathcal{I}_{1} \},  \mathbf{T}(\underbar{X})  \right)$ is called a \emph{secret key}.  Naturally, a secret key should be equiprobable on its domain to make guessing difficult, thus motivating requirement (ii).

It is not immediately clear whether any protocol $\prot$ generates a SK. 
However, it turns out that such protocols exist in great abundance.  %
 In particular, the existence of protocols that generate a SK  depends solely on the family $\{\mathcal{I}_j\}$.  %

 \begin{theorem} \label{thm:secrecyCapacity} {\cite[Theorem 6]{bib:CourtadeWesel2013}}
For a network defined by $\{\mathcal{I}_j\}$, there exists a protocol $\mathbf{P}$ which generates a SK if and only if
\begin{align}
\left| \cup_{j} \mathcal{I}_j \right| \geq \mathsf{M}^{\star}\!\left(\{\mathcal{I}_j\}\right)+1. \label{eqn:nonzeroCSK}
\end{align}
 \end{theorem}
Despite the fact that  $\mathsf{M}^{\star}\!\left(\{\mathcal{I}_j\}\right)$ corresponds to the optimal value of an ILP, it can be computed in time  polynomial in the number of messages $m=\left| \cup_{j} \mathcal{I}_j \right|$ (see \cite{bib:Milosavljevic,bib:CourtadeWesel2013}).  Therefore, for any family $\{\mathcal{I}_j\}$, we can efficiently test whether \eqref{eqn:nonzeroCSK} holds.
Hence, the essential remaining question is: ``\emph{How many transmissions are needed to generate a SK?}" 

To this end,  let $\mathcal{P}(\{\mathcal{I}_j\})$  denote the set of protocols for $\{\mathcal{I}_j\}$ that generate a SK, and define
\begin{align}
\mathsf{S}\!\left(\{\mathcal{I}_j\} \right) \triangleq \min\Big\{ \|\mathbf{T}(\underbar{X},\prot) \| : \prot \in \mathcal{P}(\{\mathcal{I}_j\})\Big\}.
\label{Sdefn}
\end{align}
That is, $\mathsf{S}\!\left(\{\mathcal{I}_j\} \right)$ is the minimum number of  transmissions needed to generate a SK.  Similarly, let $\mathcal{P}_{\mathrm{L}}(\{\mathcal{I}_j\})$  denote the set of linear protocols for $\{\mathcal{I}_j\}$ that generate a SK, and define
\begin{align}
\mathsf{S}_{\mathrm{L}}\!\left(\{\mathcal{I}_j\} \right) \triangleq \min\Big\{ \|\mathbf{T}(\underbar{X},\prot) \| : \prot \in \mathcal{P}_{\mathrm{L}}(\{\mathcal{I}_j\})\Big\}. \label{SLdefn}
\end{align}
In words, $\mathsf{S}_{\mathrm{L}}\!\left(\{\mathcal{I}_j\} \right)$ is the minimum number of  transmissions required to generate a SK when we restrict our attention to linear protocols. If $\{\mathcal{I}_j\}$ does not satisfy \eqref{eqn:nonzeroCSK}, then we set $\mathsf{S}\!\left(\{\mathcal{I}_j\} \right)=\mathsf{S}_{\mathrm{L}}\!\left(\{\mathcal{I}_j\} \right)=\infty$.

\begin{remark}
We will often write ``$\{\mathcal{I}_j\}$ generates a SK" instead of the more accurate, but cumbersome, ``For the network defined by $\{\mathcal{I}_j\}$, there exists a protocol $\mathbf{P}$ which generates a SK" whenever \eqref{eqn:nonzeroCSK} holds.
\end{remark}

\section{Generating a Single Secret Key}\label{sec:Results}
In this section, we investigate the number of transmissions required to generate a SK.  In particular, we completely characterize  $\mathsf{S}_{\mathrm{L}}\!\left(\{\mathcal{I}_j\} \right)$, and make progress toward characterizing $\mathsf{S}\!\left(\{\mathcal{I}_j\} \right)$.  We will treat the more general case of generating multiple secret keys with minimum public communication in Section \ref{sec:multipleSK}.  Since the single-SK setting is arguably the most important in practice and the notation is less cumbersome than the general case, we find it beneficial to highlight the single-SK setting in the present section.

As demonstrated in the previous section, the CCDE and SK-generation problems are closely connected through the quantity $\mathsf{M}^{\star}\!\left(\{\mathcal{I}_j\}\right)$.  Since Theorem \ref{thm:UR} and the tractability of ILP \eqref{ILP} essentially resolve the CCDE problem, it is natural to conjecture that a similar result should hold for $\mathsf{S}\!\left(\{\mathcal{I}_j\} \right)$ and $\mathsf{S}_{\mathrm{L}}\!\left(\{\mathcal{I}_j\} \right)$.  Unfortunately, there is a fundamental difference between the problems, which is revealed by the following two negative results:%

\begin{theorem}\label{thm:NPHard}
Computing $\mathsf{S}_{\mathrm{L}}\!\left(\{\mathcal{I}_j\} \right)$ is \textsf{NP}-hard.
\end{theorem}
\begin{theorem}\label{thm:SLneqS}
For any integer $k$, there exist families $\{\mathcal{I}_j\}$ for which $\mathsf{S}_{\mathrm{L}}\!\left(\{\mathcal{I}_j\} \right) > \mathsf{S}\!\left(\{\mathcal{I}_j\} \right)+k$.
\end{theorem}

For the  CCDE problem, Theorem \ref{thm:UR} asserts that linear protocols achieve optimal performance.  Furthermore,  the number of transmissions required by  linear protocols is easily computed.  For the problem  of SK generation, the opposite is true. That is, linear protocols can be suboptimal, and the number of transmissions required by  linear protocols is generally difficult to compute. 
This situation is parallel to that of  multicast network coding  and  index coding.  The two problems are closely related (cf. \cite{effros2012equivalence}), but exhibit the same dichotomy.  See 
\cite{alon2008broadcasting, lubetzky2009nonlinear, blasiak2011lexicographic} and our remark at the end of this section for more details.

\subsection{Proof of Theorem \ref{thm:NPHard}}

Despite the negative results offered by Theorems \ref{thm:NPHard} and \ref{thm:SLneqS}, we can characterize several properties of 
 $\mathsf{S}_{\mathrm{L}}\!\left(\{\mathcal{I}_j\} \right)$, $\mathsf{S}\!\left(\{\mathcal{I}_j\} \right)$, and $\mathsf{M}^{\star}\!\left(\{\mathcal{I}_j\}\right)$.  Some of these properties are demonstrated in the following results, which are needed as we progress toward proving Theorem \ref{thm:NPHard}.  A complete characterization of $\mathsf{S}_{\mathrm{L}}\!\left(\{\mathcal{I}_j\} \right)$ will be given in Theorem \ref{thm:SLminimumOneCritical}.  %
 
 \begin{lemma}\label{lem:SLleqM}
If $\{\mathcal{I}_j\}$ generates a SK, then 
\begin{align}
\mathsf{S}\!\left(\{\mathcal{I}_j\} \right) \leq \mathsf{S}_{\mathrm{L}}\!\left(\{\mathcal{I}_j\} \right)  \leq \mathsf{M}^{\star}\!\left(\{\mathcal{I}_j\}\right).
\end{align} 
\end{lemma}
\begin{proof}
By definition, $\mathsf{S}\!\left(\{\mathcal{I}_j\} \right) \leq \mathsf{S}_{\mathrm{L}}\!\left(\{\mathcal{I}_j\} \right)$ since $\mathcal{P}_{\mathrm{L}}(\{\mathcal{I}_j\}) \subseteq \mathcal{P}(\{\mathcal{I}_j\})$.  The second inequality follows from the proof of \cite[Theorem 6]{bib:CourtadeWesel2013}, in which a  linear transmission protocol $\prot$ is constructed that generates a SK with $\|\mathbf{T}(\underbar{X}),\prot)\| = \mathsf{M}^{\star}\!\left(\{\mathcal{I}_j\}\right)$ communication rounds.
\end{proof}

 We say that $\{\mathcal{J}_j \}$ is a \emph{subfamily} of $\{\mathcal{I}_j \}$ if there is a set $\mathcal{S} \subset \cup_{j} \mathcal{I}_j$ such that $\mathcal{J}_j = \mathcal{I}_j\backslash \mathcal{S}$ for all $j\in[n]$. 

 \begin{lemma}\label{lem:monotone}
If $\{\mathcal{J}_j \}$ is a subfamily of $\{\mathcal{I}_j \}$, then
\begin{align}
\mathsf{M}^{\star}\!\left(\{\mathcal{J}_j\}\right) &\leq \mathsf{M}^{\star}\!\left(\{\mathcal{I}_j\}\right),\label{MonotoneM}\\
\mathsf{S}\!\left(\{\mathcal{J}_j\} \right) &\geq \mathsf{S}\!\left(\{\mathcal{I}_j\} \right), \mbox{~and}\label{monotonS}\\
\mathsf{S}_{\mathrm{L}}\!\left(\{\mathcal{J}_j\} \right) &\geq \mathsf{S}_{\mathrm{L}}\!\left(\{\mathcal{I}_j\} \right).\label{monotonSL}
\end{align}
\end{lemma}
\begin{proof}
By De Morgan's law, it is easy to verify that 
\begin{align}
\left|\bigcap_{j\in \bar{\mathcal{S}} } \bar{\mathcal{I}_j} \right| \geq \left|\bigcap_{j\in \bar{\mathcal{S}} } \bar{\mathcal{J}_j} \right| \mbox{~~for all nonempty $\mathcal{S}\subset [n]$},
\end{align}
where $\bar{\mathcal{I}_i}\triangleq \left( \cup_j \mathcal{I}_j \right)\backslash \mathcal{I}_i$ and $\bar{\mathcal{J}_i}\triangleq \left( \cup_j \mathcal{J}_j \right)\backslash \mathcal{J}_i$.  Therefore, the constraints in ILP \eqref{ILP} are relaxed, and $\mathsf{M}^{\star}\!\left(\{\mathcal{J}_j\}\right) \leq \mathsf{M}^{\star}\!\left(\{\mathcal{I}_j\}\right)$ by definition.  

To show \eqref{monotonS}, observe that any transmission protocol which generates a SK for the subfamily $\{\mathcal{J}_j \}$ also generates a SK for the family $\{\mathcal{I}_j \}$ by ignoring the set of messages 
$\{ X_i : i \notin \cup_j \mathcal{J}_j\}$.
Hence, it follows  that $\mathsf{S}\!\left(\{\mathcal{J}_j\} \right) \geq \mathsf{S}\!\left(\{\mathcal{I}_j\} \right)$. If $\{\mathcal{J}_j \}$ can not generate a SK, the inequality trivially holds.  This argument also proves \eqref{monotonSL}.
\end{proof}

 Lemma \ref{lem:monotone} demonstrates monotonicity, but offers no insight into whether  inequalities \eqref{MonotoneM}-\eqref{monotonSL} are tight. 
 The following lemma identifies settings under which \eqref{monotonSL} holds with equality, and will prove useful later on.
 \begin{lemma}\label{lem:LinearRemoval}
 If $\mathsf{S}_{\mathrm{L}}\!\left(\{\mathcal{I}_j\} \right) < \mathsf{M}^{\star}\!\left(\{\mathcal{I}_j\}\right)$, then there exists some $\ell \in \cup_j \mathcal{I}_j$ for which 
 $\mathsf{S}_{\mathrm{L}}\!\left(\{\mathcal{I}_j - \ell\} \right) = \mathsf{S}_{\mathrm{L}}\!\left(\{\mathcal{I}_j\} \right)$.
 \end{lemma}
 \begin{proof}
Define $m \triangleq |\cup_j \mathcal{I}_j|$.  
 By definition, there is a linear transmission protocol $\prot_{\mathrm{L}}$ which generates a SK in $\mathsf{S}_{\mathrm{L}}\!\left(\{\mathcal{I}_j\} \right)$ communication rounds.  Let  $\mathbf{T}(\underbar{X},\prot_{\mathrm{L}}) = A\underbar{X}$ be the sequence of transmissions made by $\prot_{\mathrm{L}}$, and let $\{\mathsf{k}_1,\dots,  \mathsf{k}_n\}$ be valid decoding functions.

Since $\|\mathbf{T}(\underbar{X},\prot_{\mathrm{L}})\| = \mathsf{S}_{\mathrm{L}}\!\left(\{\mathcal{I}_j\} \right) < \mathsf{M}^{\star}\!\left(\{\mathcal{I}_j\}\right)$, Theorem \ref{thm:UR} asserts that the protocol $\prot_{\mathrm{L}}$ can not achieve omniscience.  Therefore, by a possible permutation of clients, we can assume without loss of generality that there is no function $\mathsf{g}_1$ for which 
\begin{align}
 \mathsf{g}_{1}\!\left(\{X_i : i\in \mathcal{I}_{1} \},   A \underbar{X}  \right) = \underbar{X} \mbox{~~with probability 1}.
\end{align}
As a consequence, there must exist a nonzero vector $\underbar{v}$ such that $A\underbar{v}=0$, and $v_i=0$ for all $i\in \mathcal{I}_1$.  Indeed, if there is no such $\underbar{v}$, then $A$ has empty nullspace and  client $c_1$ can solve a full-rank system of equations to recover $\underbar{X}$, yielding a contradiction.  Since $\underbar{v}$ is not identically zero, there is some $\ell\notin \mathcal{I}_1$ for which $v_{\ell}\neq 0$. %

Considering any such $\ell$, we define $\hat{X}_{\ell} \equiv 0$, and $\hat{X}_i \triangleq X_i$ for $i\in \cup_j (\mathcal{I}_j-\ell)$.  Also, define vectors $\hat{\underbar{X}} \triangleq [\hat{X}_1,\hat{X}_2,\dots,\hat{X}_m]^T$ and $\underbar{X}' \triangleq \hat{\underbar{X}} + X_{\ell} \cdot \underbar{v}$.  First, we note that  
\begin{align}
&\mathsf{k}_{j}\!\left(\{\hat{X}_i  : i\in \mathcal{I}_{j} \},  A \hat{\underbar{X}}  \right) =\mathsf{k}_{1}\!\left(\{\hat{X}_i  : i\in \mathcal{I}_{1} \},  A \hat{\underbar{X}}  \right) \notag%
\end{align}
for all $j\in[n]$ since 
\begin{align}
\mathsf{k}_{j}\!\left(\{X_i   : i\in \mathcal{I}_{j} \},  A \underbar{X}  \right) \notag=\mathsf{k}_{1}\!\left(\{X_i  : i\in \mathcal{I}_{1} \},  A \underbar{X}  \right) %
\end{align} with probability 1
by definition.  

Next, observe that:
\begin{align}
I\left(\mathsf{k}_{1}\!\left(\{\hat{X}_i : i\in \mathcal{I}_{1} \},  A \hat{\underbar{X}}  \right) ; A \hat{\underbar{X}} \right)
&=I\left(\mathsf{k}_{1}\!\left(\{\hat{X}_i + X_{\ell} \cdot v_i : i\in \mathcal{I}_{1} \},  A \underbar{X}'   \right) ; A \underbar{X}'  \right)\label{ident1}\\
&=I\left(\mathsf{k}_{1}\!\left(\{X_i  : i\in \mathcal{I}_{1} \},  A \underbar{X}   \right) ; A \underbar{X}  \right)=0\label{ident2}%
\end{align} 
In the above,
\begin{itemize}
\item \eqref{ident1} follows since $A \underbar{X}' = A (\hat{\underbar{X}}+X_{\ell}\cdot \underbar{v}) = A \hat{\underbar{X}}$, and $v_i=0$ for all $i\in \mathcal{I}_1$.
\item \eqref{ident2} follows from the (crucial) observation that  $\underbar{X}'$ and $\underbar{X}$ are equal in distribution, and by definition of $A$ and $\mathsf{k}_1$.
\end{itemize}

Finally, by similar reasoning, we note that the random variable $\mathsf{k}_{1}\!\left(\{\hat{X}_i : i\in \mathcal{I}_{1} \},  A \hat{\underbar{X}}  \right)$ is equiprobable on $\mathbb{F}$ since 
\begin{align}
\mathsf{k}_{1}\!\left(\{\hat{X}_i : i\in \mathcal{I}_{1} \},  A \hat{\underbar{X}}  \right)
&=\mathsf{k}_{1}\!\left(\{\hat{X}_i + X_{\ell} \cdot v_i : i\in \mathcal{I}_{1} \},  A \underbar{X}'   \right) \\
&\overset{d}{=}\mathsf{k}_{1}\!\left(\{{X}_i  : i\in \mathcal{I}_{1} \},  A \underbar{X}   \right),\label{equalInDist}
\end{align}
and $\mathsf{k}_{1}\!\left(\{{X}_i  : i\in \mathcal{I}_{1} \},  A \underbar{X}   \right)$ is equiprobable on $\mathbb{F}$ by definition.  In \eqref{equalInDist}, the notation $\overset{d}{=}$ indicates equality in distribution.

Therefore, we can conclude that a SK can be generated by the subfamily $\{\mathcal{I}_j- \ell\}_j$ %
by applying the protocol $\prot_{\mathrm{L}}$ and fixing $X_{\ell} \equiv 0$.  This proves that $\mathsf{S}_{\mathrm{L}}\!\left(\{\mathcal{I}_j - \ell\} \right) \leq \mathsf{S}_{\mathrm{L}}\!\left(\{\mathcal{I}_j\} \right)$.  By  Lemma \ref{lem:monotone}, the reverse inequality also holds.%
\end{proof}

In order to proceed, we will need to introduce  \emph{critical} families.
To this end, let $\tau\geq1$ be an integer.  A family $\{\mathcal{I}_j\}$ is \emph{$\tau$-critical} if the following hold:
\begin{enumerate}
\item[(i)] $\left| \cup_{j} \mathcal{I}_j \right| - \mathsf{M}^{\star}\!\left(\{\mathcal{I}_j\}\right) = \tau$, and
\item[(ii)] $ \mathsf{M}^{\star}\!\left(\{\mathcal{I}_j-i\}\right)= \mathsf{M}^{\star}\!\left(\{\mathcal{I}_j\}\right)$
for all $i\in  \cup_j \mathcal{I}_j$.
\end{enumerate}
 It is interesting to note that $\tau$-criticality of $\{\mathcal{I}_j\}$ can be efficiently tested since $\mathsf{M}^{\star}\!\left(\{\mathcal{I}_j\}\right)$ is computable in polynomial time.  Observe that $1$-critical families enjoy a threshold property: families $\{\mathcal{I}_j\}$ that are $1$-critical generate a secret key, and  no proper subfamilies of  $\{\mathcal{I}_j \}$ generate SKs.  This is a consequence of Theorem \ref{thm:secrecyCapacity} and the definition of $1$-criticality.%

A \emph{minimum $\tau$-critical subfamily} $\{\mathcal{J}_j^{\star}\}$ of $\{\mathcal{I}_j\}$ satisfies
\begin{align}
\left|\cup_j \mathcal{J}^{\star}_j \right| \leq \left|\cup_j \mathcal{J}_j \right|  \label{cardRelation}
\end{align}
for all other $\tau$-critical subfamilies $\{\mathcal{J}_j\}$ of $\{\mathcal{I}_j\}$.  Note that if $\{\mathcal{I}_j\}$ is $1$-critical, then $\{\mathcal{I}_j\}$ is its own unique minimum $1$-critical subfamily.

The following Theorem demonstrates that {minimum $1$-critical subfamilies} completely characterize $\mathsf{S}_{\mathrm{L}}\!\left(\{\mathcal{I}_j\} \right)$.
\begin{theorem}\label{thm:SLminimumOneCritical}
If $\{\mathcal{I}_j\}$ generates a SK, then 
\begin{align}
\mathsf{S}_{\mathrm{L}}\!\left(\{\mathcal{I}_j\} \right) = \mathsf{M}^{\star}\!\left(\{\mathcal{J}^{\star}_j\}\right) = \left|\cup_j \mathcal{J}^{\star}_j \right| - 1,
\end{align}
 where $\{\mathcal{J}_j^{\star}\}$ is a minimum $1$-critical subfamily of $\{\mathcal{I}_j\}$.
\end{theorem}
\begin{proof}
By inductively applying Lemma \ref{lem:LinearRemoval}, we can find a subfamily $\{\mathcal{T}_j\}$ of $\{\mathcal{I}_j\}$ for which 
\begin{align}
\mathsf{S}_{\mathrm{L}}\!\left(\{\mathcal{I}_j\} \right) = \mathsf{M}^{\star}\!\left(\{\mathcal{T}_j\}\right).\label{inductiveStep}
\end{align}
Let $\{\mathcal{J}_j \}$ be any $1$-critical subfamily of $\{\mathcal{T}_j\}$. We have the following chain of inequalities
\begin{align}
\mathsf{S}_{\mathrm{L}}\!\left(\{\mathcal{I}_j\} \right) \leq \mathsf{S}_{\mathrm{L}}\!\left(\{\mathcal{J}^{\star}_j\} \right)&\leq \mathsf{M}^{\star}\!\left(\{\mathcal{J}^{\star}_j\}\right) \label{step2}\\
&\leq \mathsf{M}^{\star}\!\left(\{\mathcal{J}_j\}\right) \label{step3}\\
&\leq \mathsf{M}^{\star}\!\left(\{\mathcal{T}_j\}\right)\label{step4}\\
&=\mathsf{S}_{\mathrm{L}}\!\left(\{\mathcal{I}_j\} \right).\label{step5}
\end{align}
The above steps can be justified as follows:
\begin{itemize}
\item \eqref{step2} follows from Lemmas \ref{lem:SLleqM} and \ref{lem:monotone}.
\item By definition of $\tau$-criticality, \eqref{cardRelation} is equivalent to 
$\mathsf{M}^{\star}\!\left(\{\mathcal{J}^{\star}_j\}\right) \leq \mathsf{M}^{\star}\!\left(\{\mathcal{J}_j\}\right)$. Thus, \eqref{step3} follows since $\{\mathcal{J}_j^{\star}\}$ is a minimum $1$-critical subfamily of $\{\mathcal{I}_j\}$, and $\{\mathcal{J}_j\}$ is a $1$-critical subfamily of $\{\mathcal{I}_j\}$. 
\item \eqref{step4} follows from Lemma \ref{lem:monotone}.
\item \eqref{step5} is the assertion of \eqref{inductiveStep}.
\end{itemize}
This proves that $\mathsf{S}_{\mathrm{L}}\!\left(\{\mathcal{I}_j\} \right) = \mathsf{M}^{\star}\!\left(\{\mathcal{J}^{\star}_j\}\right)$.  Recalling the definition of $1$-criticality completes the proof.
\end{proof}

The network defined by $\{\mathcal{I}_j\}$ has a natural representation as a hypergraph\footnote{We adopt the definition of a hypergraph that allows for repeated edges (i.e., multiple edges, with the same set of vertices, are permitted.}.  In particular, we make the following definition:
\begin{definition}
Consider a hypergraph $H=(\mathcal{V},\mathcal{E})$ with  vertex set $\mathcal{V}=\mathcal{C}$, and  edge set $\mathcal{E}=\cup_j \mathcal{I}_j$. $H$ is the \emph{hypergraph representation of $\{\mathcal{I}_j\}$} iff it has the following property: a vertex $c_j \in \mathcal{V}$ is contained in the edge $e \in \mathcal{E}$ if and only if $e \in \mathcal{I}_j$. 
\end{definition}

Theorem \ref{thm:SLminimumOneCritical} implies that $\mathsf{S}_{\mathrm{L}}\!\left(\{\mathcal{I}_j\} \right)$ is easily computed if we can identify a {minimum $1$-critical subfamily} of $\{\mathcal{I}_j\}$.  By Theorem \ref{thm:NPHard}, we know this must be \textsf{NP}-hard.  In order to prove this to be the case, we require the following lemma which lends a hypergraph interpretation to $1$-criticality.  For a hypergraph $H=(\mathcal{V},\mathcal{E})$, an edge set $\mathcal{E}'\subseteq \mathcal{E}$ is  a \emph{minimal connected dominating edge set} if the subhypergraph $H'=(\mathcal{V},\mathcal{E}')$ is connected, and the removal of any edge from $\mathcal{E}'$ disconnects $H'$.

\begin{lemma}\label{lem:HyperGraphConnection}
Let $H=(\mathcal{V},\mathcal{E})$ be the hypergraph representation of $\{\mathcal{I}_j\}$.  $H$ is connected if and only if
\begin{align}
\mathsf{M}^{\star}\!\left(\{\mathcal{I}_j\}\right) < \left|\cup_j \mathcal{I}_j \right| .
\end{align}
In particular, $\{\mathcal{I}_j\}$ is $1$-critical if and only if $\mathcal{E}$ is a minimal connected dominating edge set. 
\end{lemma}
\begin{proof}
First, suppose $H$ is not connected.  By definition, there must exist a nontrivial partition $\mathcal{V}=(\mathcal{S},\bar{\mathcal{S}})$ such that there is no edge $e\in \mathcal{E}$ which contains vertices from both $\mathcal{S}$ and 
$\bar{\mathcal{S}}$. Stated another way, $(\cup_{j\in \mathcal{S}} \mathcal{I}_j) \cap (\cup_{j\in \bar{\mathcal{S}}} \mathcal{I}_j) = \emptyset$.  Hence, ILP \eqref{ILP} includes the two constraints
\begin{align}
\sum_{j\in\mathcal{S}} a_j  &\geq \left|\bigcap_{j\in \bar{\mathcal{S}} } \bar{\mathcal{I}_j} \right| = \left|\bigcup_{j\in {\mathcal{S}} } {\mathcal{I}_j} \right|  \\
\sum_{j\in\bar{\mathcal{S}}} a_j  &\geq \left|\bigcap_{j\in {\mathcal{S}} } \bar{\mathcal{I}_j} \right|   = \left|\bigcup_{j\in \bar{\mathcal{S}} } {\mathcal{I}_j} \right|, 
\end{align}
the sum of which imply $\mathsf{M}^{\star}\!\left(\{\mathcal{I}_j\}\right) \geq |\cup_j \mathcal{I}_j |$.  By taking the contrapositive, we have proven 
\begin{align}
\mathsf{M}^{\star}\!\left(\{\mathcal{I}_j\}\right) < \left|\cup_j \mathcal{I}_j \right| ~~\Longrightarrow~~ H \mbox{~is connected.}
\end{align}

Next, suppose $H$ is connected, and assume without loss of generality that $\mathcal{E} = \cup_j \mathcal{I}_j  \triangleq \{1,2,\dots,m\}$.  Since $H$ is connected, there is a transmission protocol for which the entries of $\mathbf{T}(\underbar{X})$ are precisely $\{X_1 + X_j \}_{j=2}^m$. Indeed, by connectivity of $H$,  there must be some client $c$ initially holding $X_1$ and some $X_e$ (say, $X_2$ without loss of generality), and can therefore transmit $X_1+X_2$ during the first communication round.  By induction, assume that $\{X_1\!+\!X_j\}_{j=2}^{m-1}$ are transmitted during the first $m-2$ communication rounds (permuting indices of the $X_i$'s if necessary).  Again, by connectivity of $H$, there must be a client $c'$ which initially holds $X_m$ and $X_k$, where $k<m$.  Hence, in communication round $m-1$, client $c'$ can transmit $(X_1\!+\!X_k)-(X_k\!-\!X_m) = X_1\!+\!X_m$.
Noting that 
\begin{align}
(X_1,X_1\!+\!X_2, \dots, X_1\!+\!X_m)\overset{d}{=}(X_1,X_2, \dots, X_m),\notag
\end{align}
we have  $I(X_1;\mathbf{T}(\underbar{X}))=0$.  If client $c \in e\in \mathcal{E}$, then it can recover $X_1$ from the transmission $X_1+X_e$ by simply subtracting $X_e$.  Since $H$ is connected, each $c\in \mathcal{V}$ belongs to some edge in $\mathcal{E}$, and therefore all clients can recover $X_1$ losslessly.  Since $X_1$ is equiprobable on $\mathbb{F}$ by definition, we can conclude that $\{\mathcal{I}_j\}$ generates a SK.  Theorem \ref{thm:secrecyCapacity} asserts that we must have $\mathsf{M}^{\star}\!\left(\{\mathcal{I}_j\}\right) < \left|\cup_j \mathcal{I}_j \right|$, and we have proven
\begin{align}
\mathsf{M}^{\star}\!\left(\{\mathcal{I}_j\}\right) < \left|\cup_j \mathcal{I}_j \right| ~~\Longleftrightarrow~~ H \mbox{~is connected.}\label{IFFstatement}
\end{align}

We now prove the second claim.  To this end, suppose $\{\mathcal{I}_j\}$ is $1$-critical.  Then $\mathsf{M}^{\star}\!\left(\{\mathcal{I}_j\}\right) = \left|\cup_j \mathcal{I}_j \right|-1$, which implies $H$ is connected (and thus $\mathcal{E}$ is dominating) by \eqref{IFFstatement}.  Consider the subhypergraph $H'=(\mathcal{V},\mathcal{E}\backslash\{e\})$, which corresponds to the subfamily $\{\mathcal{I}_j-e\}$ of $\{\mathcal{I}_j\}$.  Since $\{\mathcal{I}_j\}$ is $1$-critical, we must have $\mathsf{M}^{\star}\!\left(\{\mathcal{I}_j-e\}\right) = \mathsf{M}^{\star}\!\left(\{\mathcal{I}_j\}\right)  = \left|\cup_j \mathcal{I}_j \right|-1 = \left|\cup_j (\mathcal{I}_j-e) \right|$.  By \eqref{IFFstatement}, $H'$ must be disconnected, and therefore $\mathcal{E}$ is a minimal connected dominating edge set. 

On the other hand, suppose $\mathcal{E}$ is a minimal connected dominating edge set.  Since $H$ is connected, \eqref{IFFstatement} implies
\begin{align}
\mathsf{M}^{\star}\!\left(\{\mathcal{I}_j\}\right) \leq \left|\cup_j \mathcal{I}_j \right|-1.
\end{align}
Since $\mathcal{E}$ is minimal, for any $e\in \mathcal{E}$, $H'=(\mathcal{V},\mathcal{E}\backslash\{e\})$ is disconnected, and  \eqref{IFFstatement} implies
\begin{align}
\mathsf{M}^{\star}\!\left(\{\mathcal{I}_j-e\}\right) \geq \left|\cup_j (\mathcal{I}_j-e) \right| =\left|\cup_j \mathcal{I}_j \right|-1.
\end{align}
Applying Lemma \ref{lem:monotone}, we must have $\mathsf{M}^{\star}\!\left(\{\mathcal{I}_j\}\right)=\mathsf{M}^{\star}\!\left(\{\mathcal{I}_j-e\}\right)$, and  $\left|\cup_j \mathcal{I}_j \right| - \mathsf{M}^{\star}\!\left(\{\mathcal{I}_j\}\right)=1$, which implies $\{\mathcal{I}_j\}$ is $1$-critical. 
 \end{proof}
 We are finally in a position to prove Theorem \ref{thm:NPHard}.
 \begin{proof}[Proof of Theorem \ref{thm:NPHard}]
Let $H=(\mathcal{V},\mathcal{E})$ be the hypergraph representation of $\{\mathcal{I}_j\}$. 
We can assume $\{\mathcal{I}_j\}$ generates a SK.
By Theorem \ref{thm:SLminimumOneCritical} and Lemma \ref{lem:HyperGraphConnection}, computing $\mathsf{S}_{\mathrm{L}}\!\left(\{\mathcal{I}_j\} \right)$ is equivalent to computing the the number of edges in  a minimum connected dominating edge set (i.e., a minimal connected dominating edge set with fewest possible edges).  It is easy to see that the \textsf{NP}-complete {\sc Set Cover Decision Problem} is a special case.  

Indeed, consider any subsets $\mathcal{A}_1,\mathcal{A}_2,\dots,\mathcal{A}_k$ whose union covers a finite set $\mathcal{U}$.  For $u'\notin \mathcal{U}$,  define $\mathcal{U}' = \mathcal{U}\cup\{u'\}$, and $\mathcal{A}_j' = \mathcal{A}_j\cup \{u'\}$ for $j\in[k]$.   Clearly, $\{\mathcal{A}_{j_i}\}_{i=1}^m$ is a minimum  cover of $\mathcal{U}$ if and only if  $\{\mathcal{A}'_{j_i}\}_{i=1}^m$ is a minimum  connected cover of $\mathcal{U}'$.
\end{proof}

\begin{remark}
Together, Theorem \ref{thm:SLminimumOneCritical} and Lemma \ref{lem:HyperGraphConnection} give a succinct characterization of $\mathsf{S}_{\mathrm{L}}\!\left(\{\mathcal{I}_j\} \right)$ in terms of hypergraph connectivity.  We extend this result to  the generation of multiple secret keys at the end of Section \ref{sec:multipleSK} using a stronger form of hypergraph connectivity.
\end{remark}

\subsection{Proof of Theorem \ref{thm:SLneqS} }

Before proving Theorem \ref{thm:SLneqS}, consider  the following constructive example: 
Let $n=7$, and consider the family $\{\mathcal{I}_j\}$ defined by $\mathcal{I}_1 = \{1,2,3,4\}$, and $\mathcal{I}_2, \dots ,\mathcal{I}_7$ are all ${4 \choose 2}$ distinct $2$-element subsets of $\{1,2,3,4\}$.  By direct computation, we find that $\{\mathcal{I}_j-\{1\}\}$ is a minimum $1$-critical subfamily, and hence
$\mathsf{S}_{\mathrm{L}}\!\left(\{\mathcal{I}_j\} \right) =   2$ by Theorem \ref{thm:SLminimumOneCritical}.
Suppose $\mathbb{F}=\{0,1,\alpha,\beta\}^2 = \mathsf{GF}(4)\!\times \!\mathsf{GF}(4)$.  Thus, we can express $X_j = (X_j^{(1)}, X_j^{(2)})$ for each $j=1,\dots,4$, where $X_j^{(1)}, X_j^{(2)}$ are mutually independent, each equiprobable on $\mathsf{GF}(4)$.  It is readily verified that the single transmission 
\begin{align}
\left(X_1^{(1)}\!\!+\! \alpha X_2^{(1)}\!\!+\!X_3^{(1)}, X_1^{(1)}\!\!+\!\beta X_2^{(1)}\!\!+\! X_4^{(1)}\right) \in \mathbb{F} 
\end{align}
by client $c_1$ permits reconstruction of the SK 
\begin{align}
\mathsf{k}_{1}\!\left(\{X_i  : i\in \mathcal{I}_{1} \},  \mathbf{T}(\underbar{X})   \right)= (X_3^{(1)},X_4^{(1)})\in \mathbb{F}
\end{align}
 at all clients.  Hence, we can conclude
$1 = \mathsf{S}\!\left(\{\mathcal{I}_j\} \right)  < \mathsf{S}_{\mathrm{L}}\!\left(\{\mathcal{I}_j\} \right) = \mathsf{M}^{\star}\!\left(\{\mathcal{I}_j\}\right) = 2$.

The above construction is a \emph{vector-linear} transmission protocol, and cannot be realized by a protocol which is linear over $\mathbb{F}$.  A natural question is whether it is possible to bound the gap between $\mathsf{S}\!\left(\{\mathcal{I}_j\} \right)$ and $\mathsf{S}_{\mathrm{L}}\!\left(\{\mathcal{I}_j\} \right)$.  As asserted by Theorem \ref{thm:SLneqS}, the answer to this is negative.  Indeed,  it is straightforward to generalize the previous construction and make the gap arbitrarily large.  

To this end, consider a network of $n=\binom{m}{2}+1$ clients such that $\mathcal{I}_1=[m]$ and the other $\binom{m}{2}$ clients possess distinct pairs of messages. Observe that the $1$-critical subfamilies of $\{\mathcal{I}_j\}$ are obtained by removing a single message -- i.e., if any two messages $m_1,m_2\in[m]$ are removed then the resulting hypergraph representation of $\{\mathcal{I}_j-\{m_1,m_2\}\}$ is no longer connected. This implies that $\mathsf{S}_{\mathrm{L}}\!\left(\{\mathcal{I}_j\} \right)=m-2$. To show that there exists a nonlinear scheme that can do better, we show that $\mathsf{M}^\star(\{\mathcal{I}_j\}) =  m-2$:
\begin{itemize}
\item To show that $\mathsf{M}^\star(\{\mathcal{I}_j\})\leq m-2$, let client $c_1$ transmit $m-2$ independent linear combinations of the messages. Provided the encoding matrix $A$ is full rank (e.g., a Vandermonde matrix), every other node can use its own pair of messages to recover the other $m-2$.
\item We note that $\mathsf{M}^\star(\{\mathcal{I}_j\})\geq \mathsf{S}_{\mathrm{L}}\!\left(\{\mathcal{I}_j\} \right) = m-2$ by Lemma \ref{lem:SLleqM}, and therefore $\mathsf{M}^\star(\{\mathcal{I}_j\})=m-2$ as claimed.
\end{itemize}
Now, we simply split the packets and apply the optimal transmission protocol over the first halves of the packets as we did previously. This vector-linear scheme generates a SK with $m/2-1$ transmissions, which is an improvement of $m/2-1$ transmissions over the best linear scheme.  Since $m$ was arbitrary, we have shown that the gap between  $\mathsf{S}\!\left(\{\mathcal{I}_j\} \right)$ and $\mathsf{S}_{\mathrm{L}}\!\left(\{\mathcal{I}_j\} \right)$ cannot be bounded in general, proving Theorem \ref{thm:SLneqS}.

\begin{remark}
Our proof that $\mathsf{S}\!\left(\{\mathcal{I}_j\} \right) < \mathsf{S}_{\mathrm{L}}\!\left(\{\mathcal{I}_j\} \right)$ is similar to the index coding problem, where the suboptimality of linear schemes was also shown by demonstrating a gap between the performance of linear and vector-linear coding schemes \cite{alon2008broadcasting, lubetzky2009nonlinear}.  For several years, it was unknown whether vector-linear coding schemes were optimal in the index coding problem.  However, Blasiak et al. have since proved that even vector-linear coding is strictly suboptimal for the index coding problem \cite{blasiak2011lexicographic}.  We conjecture the same is true for the present setting.
\end{remark}

 \section{Generating Multiple Secret Keys}\label{sec:multipleSK}
Until now,  we have focused exclusively on protocols that generate a single SK.  However, it is also natural to consider protocols that generate $\tau$ independent secret keys. Indeed, the secrecy capacity as defined in \cite{bib:CsiszarNarayanIT2004} translates to the maximum number of secret keys that can possibly be generated in the combinatorial setting we consider.  Thus, it is interesting to study the tradeoff between the number of secret keys that can be generated and the number of public transmissions required to do so.

To this end, we say a transmission protocol $\mathbf{P}$ (with corresponding transmission sequence $\mathbf{T}(\underbar{X})$) generates $\tau$ secret keys if there exist decoding functions $\{\mathsf{k}_1,\mathsf{k}_2,\dots, \mathsf{k}_n\}$ which satisfy the following three properties:
\begin{enumerate}
\item[(i)] For all $j\in [n]$, and with probability 1, 
$$
\!\!\!\!\!\!\!\!\mathsf{k}_{j}\!\left(\{X_i   : i\in \mathcal{I}_{j} \},  \mathbf{T}(\underbar{X})  \right) \notag=\mathsf{k}_{1}\!\left(\{X_i  : i\in \mathcal{I}_{1} \},  \mathbf{T}(\underbar{X})   \right).
$$
\item[(ii)] $\mathsf{k}_{1}\!\left(\{X_i  : i\in \mathcal{I}_{1} \},  \mathbf{T}(\underbar{X})   \right)$ is equiprobable on $\mathbb{F}^\tau$.
\item[(iii)] $I\left(\mathsf{k}_{1}\!\left(\{X_i  : i\in \mathcal{I}_{1} \},  \mathbf{T}(\underbar{X})   \right);  \mathbf{T}(\underbar{X})  \right)=0$.
\end{enumerate} 

Note that (i)--(iii) are  the same requirements for generating a single SK with one exception: we require that $\mathsf{k}_{1}\!\left(\{X_i  : i\in \mathcal{I}_{1} \},  \mathbf{T}(\underbar{X})   \right)$ is uniformly over $\mathbb{F}^\tau$.  In other words, we require that each client  recovers $\tau$ independent SKs,  each known to all clients and private from any eavesdropper. As stated in  {\cite[Theorem 6]{bib:CourtadeWesel2013}}, Theorem \ref{thm:secrecyCapacity} can be generalized as follows: 
\begin{theorem}\label{thm:secrecyCapacityTau}
For a network defined by $\{\mathcal{I}_j\}$, there exists a protocol $\mathbf{P}$ which generates $\tau$  SKs if and only if
\begin{align}
\left| \cup_{j} \mathcal{I}_j \right| \geq \mathsf{M}^{\star}\!\left(\{\mathcal{I}_j\}\right)+\tau. \label{eqn:nonzeroCSKTau}
\end{align}
 \end{theorem}
Analogous to the definition of $\mathsf{S}_{\mathrm{L}}\!\left(\{\mathcal{I}_j\} \right)$ in  \eqref{SLdefn}, let   $\mathsf{S}^{(\tau)}_{\mathrm{L}}\!\left(\{\mathcal{I}_j\} \right)$ denote the minimum number of transmissions required by a  linear protocol to generate $\tau$ independent secret keys.  
 A minor modification of our arguments for the single-SK setting yields:
 \begin{theorem}
Let $\tau\geq 1$ be an integer.  If there is a protocol $\prot$ for $\{\mathcal{I}_j\}$ which generates $\tau$ independent secret keys, then 
\begin{align}
\mathsf{S}^{(\tau)}_{\mathrm{L}}\!\left(\{\mathcal{I}_j\} \right) = \mathsf{M}^{\star}\!\left(\{\mathcal{J}^{\star}_j\}\right) = \left|\cup_j \mathcal{J}^{\star}_j \right| - \tau,
\end{align}
 where $\{\mathcal{J}_j^{\star}\}$ is a minimum $\tau$-critical subfamily of $\{\mathcal{I}_j\}$.
 \end{theorem}
 
In the single-SK setting, Lemma \ref{lem:HyperGraphConnection} gave a succinct interpretation of minimum 1-critical subfamilies of $\{\mathcal{I}_j\}$ as connectedness of $H$, the hypergraph representation of $\{\mathcal{I}_j\}$.  When combined with Theorem \ref{thm:SLminimumOneCritical}, we find that $\mathsf{S}_{\mathrm{L}}\!\left(\{\mathcal{I}_j\} \right)$ is in one-to-one correspondence with the size  of a minimum connected dominating edge-set of $H$. The chief difficulty in giving a similarly succinct characterization of $\mathsf{S}^{(\tau)}\!\left(\{\mathcal{I}_j\} \right)$ lies in generalizing Lemma \ref{lem:HyperGraphConnection} appropriately for $\tau\geq 2$. In order to do so, we will need to introduce a more general notion of hypergraph connectivity.

There are many definitions of connectivity for hypergraphs. We recall two common examples here:  
\begin{itemize}
\item \emph{Example 1:} A hypergraph is said to be $\tau$-edge connected if the deletion of fewer than $\tau$ edges leaves $H$ connected.  
\item \emph{Example 2:} A more stringent notion of connectivity is partition-connectivity \cite{frank2003decomposing}. A hypergraph $H$ is said to be $\tau$-partition connected if for all partitions $\mathscr{P}$ of the vertex set, the number of hyperedges intersecting at least two parts of $\mathscr{P}$ is at least $\tau(|\mathscr{P}|-1)$. 
\end{itemize}
We say that a  multigraph is $\tau$-partition connected if it contains $\tau$ edge-disjoint spanning trees.  This definition is justified by recalling a classical result  of Nash-Williams \cite{nash1961edge} and Tutte  \cite{tutte1961problem}.
\begin{theorem}\label{thm:NWT}
An undirected multigraph $G = (\mathcal{V},\mathcal{E})$ contains $\tau$ edge-disjoint spanning trees iff for every partition $\mathscr{P}$ of $\mathcal{V}$ into disjoint sets  $\mathcal{V}_1, \mathcal{V}_2, \dots, \mathcal{V}_{|\mathscr{P}|}$,
\begin{align}
\sum_{e\in \mathcal{E}} (r(e\,;\mathscr{P})-1) \geq  \tau(|\mathscr{P}|-1),
\end{align}
where $r(e\,;\mathscr{P})$ is the number of parts in $\mathscr{P}$ that the edge $e$ intersects (i.e., its rank with respect to the partition $\mathscr{P}$).
\end{theorem}

\begin{definition}
A multigraph $G=(\mathcal{V},\mathcal{E}_M)$ is \emph{induced by a hypergraph} $H=(\mathcal{V},\mathcal{E})$ iff it can be decomposed into a disjoint collection of  simple graphs $\{G_e \}_{e \in \mathcal{E}}$, where $G_e=(e,E_e)$ is a connected graph on the vertex set $e\in \mathcal{E}$.%
\end{definition}
Two examples of multigraphs induced by a hypergraph are given in Figure \ref{fig:inducedH}. 

\begin{figure}
\begin{center}
\def\svgwidth{.95\textwidth}
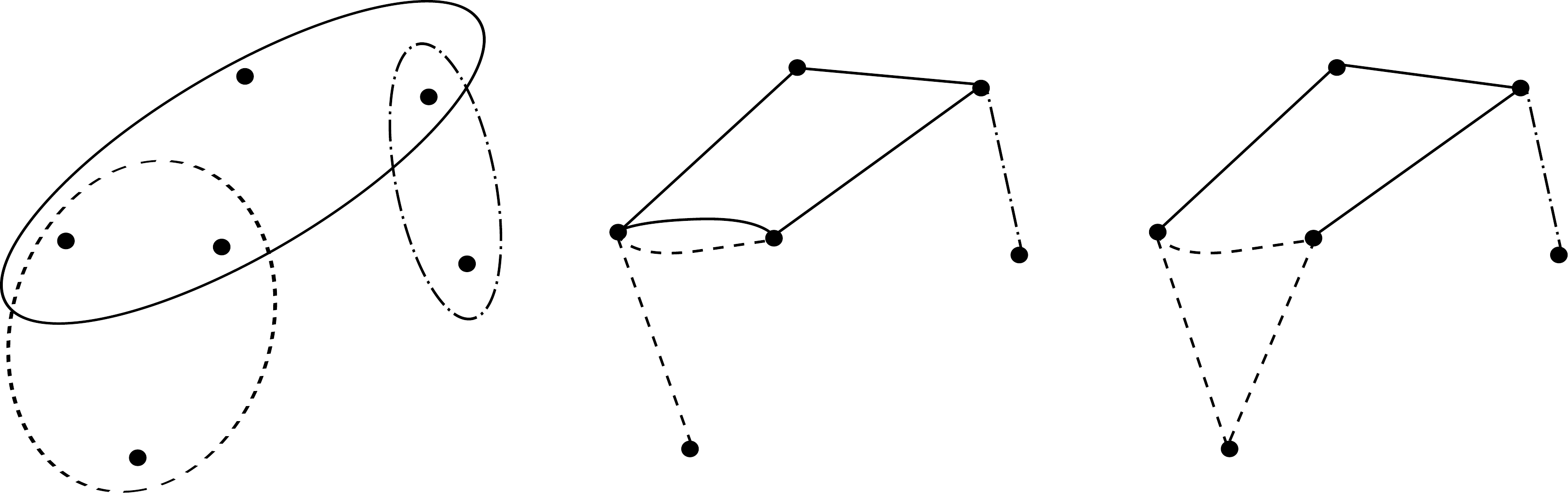
\end{center}
\caption{An example of a hypergraph $H$ (left) and two induced multigraphs (center, right).  Line textures are used to emphasize the relationship between the hypergraph edges and the decomposition of the multigraphs into corresponding simple connected graphs.}\label{fig:inducedH}
\end{figure}

\begin{figure}
\begin{center}
\def\svgwidth{.79\textwidth}
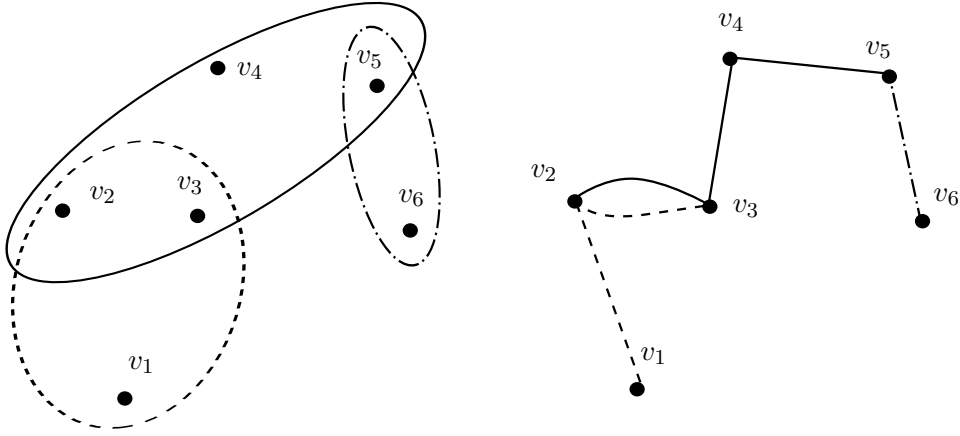
\end{center}
\caption{A hypergraph $H$ (left) and the induced multigraph $G_{H,\prec}$ (right) for the vertex-ordering $v_1 \prec v_2 \prec \cdots \prec v_6$. Line textures are used to emphasize the relationship between the hypergraph edges and the decomposition of the multigraphs into corresponding simple connected graphs.}\label{fig:Gorder}
\end{figure}

\begin{definition}
A hypergraph $H=(\mathcal{V},\mathcal{E})$ is \emph{inherently $\tau$-connected} iff every induced multigraph contains at least $\tau$ edge-disjoint spanning trees.
\end{definition}

A pleasant  generalization  of Theorem \ref{thm:NWT} holds for inherently $\tau$-connected hypergraphs.
\begin{theorem}\label{thm:inherentTau}
A hypergraph $H=(\mathcal{V},\mathcal{E})$ is inherently $\tau$-connected iff for any partition $\mathscr{P}$ of $\mathcal{V}$ into disjoint sets  $\mathcal{V}_1, \mathcal{V}_2, \dots, \mathcal{V}_{|\mathscr{P}|}$, 
\begin{align}
\sum_{e\in \mathcal{E}} (r(e\,;\mathscr{P})-1) \geq  \tau(|\mathscr{P}|-1),
\end{align}
where $r(e\,;\mathscr{P})$ is the number of parts in $\mathscr{P}$ that the hyperedge $e$ intersects.
\end{theorem}
Theorem \ref{thm:inherentTau} follows as an easy corollary of Theorem \ref{thm:NWT} and the definition of an inherently $\tau$-connected hypergraph.  However, a stronger version of Theorem \ref{thm:inherentTau} can be distilled from our proof of Lemma \ref{lem:HyperGraphConnection2}, which is stated shortly.  Specifically, we will see that a hypergraph $H$ is inherently $\tau$-connected iff a relatively small subset of induced multigraphs induced by $H$ contain $\tau$ edge-disjoint spanning trees.  For a precise statement,  see our remark following the proof of Lemma \ref{lem:HyperGraphConnection2}. Though not needed elsewhere in this paper, we remark that there is an analogous version of Theorem \ref{thm:NWT} for hypergraphs and partition-connectivity due to Frank, Kir\'{a}ly, and  Kriesell \cite{frank2003decomposing}.
\begin{theorem}\label{thm:PartitionConn}
A hypergraph $H=(\mathcal{V},\mathcal{E})$ can be decomposed into $\tau$ subhypergraphs, each of which is 1-partition connected iff for any partition $\mathscr{P}$ of $\mathcal{V}$ into disjoint sets  $\mathcal{V}_1, \mathcal{V}_2, \dots, \mathcal{V}_{|\mathscr{P}|}$, the number of hyperedges intersecting at least two parts of $\mathscr{P}$ is at least $\tau(|\mathscr{P}|-1)$ (i.e., $H$ is $\tau$-partition connected).
\end{theorem}

We point out that $\tau$-partition connectivity is a more stringent condition than inherent $\tau$-connectivity, as reflected by Theorems \ref{thm:inherentTau} and \ref{thm:PartitionConn}.

For a hypergraph $H=(\mathcal{V},\mathcal{E})$, an edge set $\mathcal{E}'\subseteq \mathcal{E}$ is a minimal inherently $\tau$-connected edge-set if the subhypergraph $H'=(\mathcal{V},\mathcal{E}')$ is inherently $\tau$-connected, and the removal of any edge from $\mathcal{E}'$ results in a subhypergraph that is not inherently $\tau$-connected. Further, define
\begin{align}
\varrho_{\tau}(H) = \min \Big\{ \left|\mathcal{E}'\right| : \mathcal{E}'\subseteq \mathcal{E} \mbox{~~is an inherently $\tau$-connected edge-set} \Big\}.
\end{align}
In other words, $\varrho_{\tau}(H)$ is the minimum number of edges in an inherently $\tau$-connected subhypergraph $H'=(\mathcal{V},\mathcal{E}')$ of $H=(\mathcal{V},\mathcal{E})$.  Note that $\varrho_{\tau}(H)$ is the minimum number of edges in a connected dominating edge set when $\tau=1$, and thus its computation is \textsf{NP}-hard in general.

\begin{lemma}\label{lem:HyperGraphConnection2}
Let  $H=(\mathcal{V},\mathcal{E})$ be the hypergraph representation of $\{\mathcal{I}_j\}$.  $H$ is inherently $\tau$-connected if and only if
\begin{align}
\mathsf{M}^{\star}\!\left(\{\mathcal{I}_j\}\right) \leq \left|\cup_j \mathcal{I}_j \right| -\tau.
\end{align}
In particular, $\{\mathcal{I}_j\}$ is $\tau$-critical if and only if $\mathcal{E}$ is a minimal inherently $\tau$-connected edge-set.
\end{lemma}

Before we begin the proof of Lemma \ref{lem:HyperGraphConnection2}, we take a moment to describe a special class of multigraphs that are induced by $H$.  %
For a hypergraph $H=(\mathcal{V},\mathcal{E})$, let $\prec$ be a strict total order on $\mathcal{V}$.  That is, if $\mathcal{V}=\{v_1,v_2, \dots, v_n\}$, there is a permutation $\pi$ on $\{1,\dots,n\}$ %
for which  $v_{\pi(1)} \prec v_{\pi(2)} \prec \cdots \prec v_{\pi(n)}$.  
Define the multigraph $G_{H,\prec}$ induced by $H$, with decomposition $\{G_e\}_{e\in\mathcal{E}}$, as follows:
For each $e\in\mathcal{E}$, let $G_e$ be a path that connects the vertices contained in $e$ in ascending order (with respect to $\prec$).  In other words, if $e=\{v_{i_1}, v_{i_2}, \dots, v_{i_k}\}$, where $v_{i_j}\prec v_{i_{\ell}}$ for $i_j < i_{\ell}$, then the edge-set of $G_e$ is precisely $\{v_{i_1},v_{i_2}\}, \{v_{i_2},v_{i_3}\}, \dots, \{v_{i_{k-1}},v_{i_{k}}\}$.
An example is shown in Figure \ref{fig:Gorder}.

\begin{proof}[Proof of Lemma \ref{lem:HyperGraphConnection2}]
Let $(a_1^{\star}, \dots, a_n^{\star})$ be an optimal solution to ILP \eqref{ILP}.  First, suppose $\mathsf{M}^{\star}\!\left(\{\mathcal{I}_j\}\right) \leq \left|\cup_j \mathcal{I}_j \right| -\tau$.  Then, for any partition $\mathscr{P} = \{\mathcal{V}_1, \mathcal{V}_2, \dots, \mathcal{V}_k\}$ of $\mathcal{V}$, we have:
\begin{align}
(\left|\cup_j \mathcal{I}_j \right| -\tau)(k-1) & \geq \mathsf{M}^{\star}\!\left(\{\mathcal{I}_j\}\right)(k-1)\\
&=\sum_{i=1}^k \left( \mathsf{M}^{\star}\!\left(\{\mathcal{I}_j\}\right) - \sum_{j\in \mathcal{V}_i} a_j^{\star} \right)\\
&=\sum_{i=1}^k  \sum_{j\in \bar{\mathcal{V}}_i} a_j^{\star} \\
&\geq \sum_{i=1}^k  \left|\bigcap_{j\in {\mathcal{V}_i} } \bar{\mathcal{I}_j} \right| \label{ILPineqStar}\\
&=k \left|\cup_j \mathcal{I}_j \right| - \sum_{i=1}^k \left|\bigcup_{j\in {\mathcal{V}_i} } {\mathcal{I}_j} \right|,
\end{align}
where \eqref{ILPineqStar} follows by feasibility of $(a_1^{\star}, \dots, a_n^{\star})$ for ILP \eqref{ILP}.
Rearranging, we find
\begin{align}
\sum_{i=1}^k \left|\bigcup_{j\in {\mathcal{V}_i} } {\mathcal{I}_j} \right| \geq \left|\cup_j \mathcal{I}_j \right| +\tau (k-1). \label{tauSpanIneq}
\end{align}
Now, let $G$ be an arbitrary multigraph induced by $H$ with decomposition  given by $\{G_e \}_{e \in \mathcal{E}}$. Note that if $e \in \mathcal{E}$ intersects $r(e\,;\mathscr{P})$ parts of the partition $\mathscr{P}$, then at least $r(e\,;\mathscr{P})-1$ edges of $G_e$  cross the partition $\mathscr{P}$.  Therefore,
\begin{align}
\Omega \left(G,\mathscr{P}\right) \geq \sum_{e\in \mathcal{E}} (r(e\,;\mathscr{P})-1) = \left(\sum_{i=1}^k \left|\bigcup_{j\in {\mathcal{V}_i} } {\mathcal{I}_j} \right|\right) - \left|\cup_j \mathcal{I}_j \right|,
\end{align}
where $\Omega\left(G,\mathscr{P}\right)$ denotes  the number of edges in $G$ that cross the partition $\mathscr{P}$. 
Since the partition $\mathscr{P}$ and induced multigraph $G$ were arbitrary, it follows from Theorem \ref{thm:NWT} and \eqref{tauSpanIneq} that  $H$ is inherently $\tau$-connected. Thus, we have shown:
\begin{align}
\mathsf{M}^{\star}\!\left(\{\mathcal{I}_j\}\right) \leq \left|\cup_j \mathcal{I}_j \right| -\tau~~\Longrightarrow~~ H \mbox{~is inherently $\tau$-connected.}
\end{align}

Next suppose $H$ is inherently $\tau$-connected.  By optimality of $(a_1^{\star}, \dots, a_n^{\star})$, there exists a partition $\mathscr{P}^{\star} = \{\mathcal{V}_1, \mathcal{V}_2, \dots, \mathcal{V}_k\}$ of $\mathcal{V}$ (see \cite[Appendix A]{bib:CourtadeWesel2013}, \cite{bib:Schrijver2003}) such that 
\begin{align}
\sum_{i=1}^k  \sum_{j\in \bar{\mathcal{V}}_i} a_j^{\star} = \sum_{i=1}^k  \left|\bigcap_{j\in {\mathcal{V}_i} } \bar{\mathcal{I}_j} \right|.
\end{align}
Now, consider an arbitrary order $\prec$ on $\mathcal{V}$ which satisfies $u \prec v$ if $u \in \mathcal{V}_i$, $v \in \mathcal{V}_j$ and $i < j$.  In this case, if $e \in \mathcal{E}$ intersects $r(e\,;\mathscr{P}^{\star})$ parts of the partition $\mathscr{P}^{\star}$, then the path in $G_{H,\prec}$ generated by the hyperedge  $e$ (i.e., $G_e$) will have precisely $r(e\,;\mathscr{P}^{\star})-1$ edges that cross $\mathscr{P}^{\star}$.  Since $H$ is inherently $\tau$-connected, we have
\begin{align}
\left(\sum_{i=1}^k \left|\bigcup_{j\in {\mathcal{V}_i} } {\mathcal{I}_j} \right|\right) - \left|\cup_j \mathcal{I}_j \right| = \sum_{e\in \mathcal{E}} (r(e\,;\mathscr{P}^{\star})-1) = \Omega \left(G_{H,\prec},\mathscr{P}^{\star}\right) \geq \tau(k-1) \label{hSpanEqn}
\end{align}
by Theorem \ref{thm:NWT}.  Proceeding in a fashion similar to before, we have for $\mathscr{P}^{\star}$ that
\begin{align}
\mathsf{M}^{\star}\!\left(\{\mathcal{I}_j\}\right)(k-1)
&=\sum_{i=1}^k \left( \mathsf{M}^{\star}\!\left(\{\mathcal{I}_j\}\right) - \sum_{j\in \mathcal{V}_i} a_j^{\star} \right)\\
&=\sum_{i=1}^k  \sum_{j\in \bar{\mathcal{V}}_i} a_j^{\star} \\
&= \sum_{i=1}^k  \left|\bigcap_{j\in {\mathcal{V}_i} } \bar{\mathcal{I}_j} \right| \\
&=k \left|\cup_j \mathcal{I}_j \right| - \sum_{i=1}^k \left|\bigcup_{j\in {\mathcal{V}_i} } {\mathcal{I}_j} \right|\\
&\leq (k-1)(\left|\cup_j \mathcal{I}_j \right| -\tau),
\end{align}
where the final inequality follows from \eqref{hSpanEqn}.  Hence, 
\begin{align}
\mathsf{M}^{\star}\!\left(\{\mathcal{I}_j\}\right) \leq \left|\cup_j \mathcal{I}_j \right| -\tau ~~\Longleftrightarrow~~ H \mbox{~is inherently $\tau$-connected.}\label{IFFstatement}
\end{align}

We now prove the second claim.  
To this end, suppose $\{\mathcal{I}_j\}$ is $\tau$-critical.  Then $\mathsf{M}^{\star}\!\left(\{\mathcal{I}_j\}\right) = \left|\cup_j \mathcal{I}_j \right|-\tau$, which implies $H$ is inherently $\tau$-connected by \eqref{IFFstatement}.  Consider the subhypergraph $H'=(\mathcal{V},\mathcal{E}\backslash\{e\})$, which corresponds to the subfamily $\{\mathcal{I}_j-e\}$ of $\{\mathcal{I}_j\}$.  Since $\{\mathcal{I}_j\}$ is $\tau$-critical, we must have $\mathsf{M}^{\star}\!\left(\{\mathcal{I}_j-e\}\right) = \mathsf{M}^{\star}\!\left(\{\mathcal{I}_j\}\right)  = \left|\cup_j \mathcal{I}_j \right|-\tau  = \left|\cup_j (\mathcal{I}_j-e) \right|  - \tau +1$.  By \eqref{IFFstatement}, $H'$ cannot be inherently $\tau$-connected, and therefore $\mathcal{E}$  is a minimal inherently $\tau$-connected edge-set.

On the other hand, suppose $\mathcal{E}$  is a minimal inherently $\tau$-connected edge-set.  Then, \eqref{IFFstatement} implies
\begin{align}
\mathsf{M}^{\star}\!\left(\{\mathcal{I}_j\}\right) \leq \left|\cup_j \mathcal{I}_j \right|-\tau.
\end{align}
Since $\mathcal{E}$  is a inherently $\tau$-connected edge-set, for any $e\in \mathcal{E}$, $H'=(\mathcal{V},\mathcal{E}\backslash\{e\})$ is not inherently $\tau$-connected, and  \eqref{IFFstatement} implies
\begin{align}
\mathsf{M}^{\star}\!\left(\{\mathcal{I}_j-e\}\right) \geq \left|\cup_j (\mathcal{I}_j-e) \right| -\tau+1 =\left|\cup_j \mathcal{I}_j \right|-\tau.
\end{align}
Applying Lemma \ref{lem:monotone}, we must have $\mathsf{M}^{\star}\!\left(\{\mathcal{I}_j\}\right)=\mathsf{M}^{\star}\!\left(\{\mathcal{I}_j-e\}\right)$, and  $\left|\cup_j \mathcal{I}_j \right| - \mathsf{M}^{\star}\!\left(\{\mathcal{I}_j\}\right)=\tau$, which implies $\{\mathcal{I}_j\}$ is $\tau$-critical. 
 \end{proof}
\begin{remark}
From the proof of Lemma \ref{lem:HyperGraphConnection2}, we observe  that a hypergraph $H$ is inherently $\tau$-connected if and only if $G_{H,\prec}$ contains $\tau$ edge-disjoint spanning trees for every strict order $\prec$.  Hence, this apparently weaker condition is, in fact,  necessary and sufficient  for {any} multigraph induced by $H$ to contain $\tau$ edge-disjoint spanning trees.
\end{remark}

In summary, we have found the following characterization of $\mathsf{S}^{(\tau)}_{\mathrm{L}}\!\left(\{\mathcal{I}_j\} \right)$:
\begin{theorem}\label{bigThm}
If $H$ is the hypergraph representation of the network defined by $\{\mathcal{I}_j\}$, then
\begin{align}
\mathsf{S}^{(\tau)}_{\mathrm{L}}\!\left(\{\mathcal{I}_j\} \right) = \varrho_{\tau}(H) - \tau.
\end{align}
\end{theorem}

When we restrict ourselves to linear protocols, Theorem \ref{bigThm} elucidates a direct correspondence between the number of public transmissions required to generate $\tau$ SKs in a network and the inherent $\tau$-connectivity of the representative hypergraph. As an illustrative example, consider the following network with 15 clients: 
\begin{example}\label{example}
 Let $\mathcal{I}_1 = \{5,7,10,11,13,14,15\}$, and let $\{ \mathcal{I}_j \}_{j=1}^{15}$ be the 14 different cyclic shifts of $\mathcal{I}_1$ (e.g., $\mathcal{I}_2=\{1,6,8,11,12,14,15\}$, $\mathcal{I}_3=\{1,2,7,9,12,13,15\}$, $\dots$).  Since the number of messages $m=15$ is modestly small, we are able to compute $\varrho_{\tau}(H)$ explicitly for the hypergraph representation of the network defined by $\{\mathcal{I}_j\}$, and therefore also $\mathsf{S}^{(\tau)}_{\mathrm{L}}\!\left(\{\mathcal{I}_j\} \right)$ by invoking Theorem \ref{bigThm}.  Below, Table \ref{tab:Stau} gives $\mathsf{S}^{(\tau)}_{\mathrm{L}}\!\left(\{\mathcal{I}_j\} \right)$ for $\tau \geq 1$:

\begin{table}[h!]
\begin{center}
\begin{tabular}{c  || *{6}{c}r}
$\tau$              & 1 & 2 & 3 & 4 & 5  & 6 & $\geq 7$ \\
\hline
$\mathsf{S}^{(\tau)}_{\mathrm{L}}\!\left(\{\mathcal{I}_j\} \right)$  & 2 & 4 & 4& 6 & 8 & 8 & $\infty$ %
\end{tabular}\caption{$\mathsf{S}^{(\tau)}_{\mathrm{L}}\!\left(\{\mathcal{I}_j\} \right)$ vs. $\tau$ for the network given in Example \ref{example}.  Note that $\mathsf{S}^{(\tau)}_{\mathrm{L}}\!\left(\{\mathcal{I}_j\} \right)=\infty$ indicates that it is not possible to generate $\tau$ secret keys with any number of transmissions.}\label{tab:Stau}
\end{center}
\end{table}
\end{example}

In another example\footnote{We remark that this generalizes a very recent result due to Mukherjee and Kashyap \cite{bib:MukherjeeKashyap2014}.}, we give a complete characterization for $\mathsf{S}^{(\tau)}_{\mathrm{L}}\!\left(\{\mathcal{I}_j\} \right)$ when each pair of clients shares a unique message (i.e., $m={n \choose 2}$, and the hypergraph representation of the network defined by $\{\mathcal{I}_j\}$ is a complete (simple) graph on $n$ vertices).  This network model was  called the PIN model by Nitinawarat and Narayan \cite{bib:NitinawaratNarayan2010}.
\begin{example}
In the PIN model, $\mathsf{S}^{(\tau)}_{\mathrm{L}}\!\left(\{\mathcal{I}_j\} \right) = \tau(n-2)$, where $1\leq \tau \leq \lfloor n/2\rfloor$.  Indeed, a simple graph is inherently $\tau$-connected iff it contains $\tau$ edge-disjoint spanning trees by Theorem \ref{thm:NWT}.  Thus, a simple counting argument gives $\varrho_{\tau}(H) =\tau(n-1)$.  An application of Theorem \ref{bigThm} proves the claim.
\end{example}

It is an interesting combinatorial design problem to specify ideal message distributions amongst clients (subject to constraints) that  allow secret-key generation with fewest transmissions.  For example, how many transmissions are required to generate a SK subject to the constraint that each message is initially held by at most $t$ clients?  This general problem is beyond the scope of the present paper, and we leave it for future work.

 \section{Concluding Remarks}\label{sec:conclusion}
 
In this paper, we have completely characterized the number of public transmissions required to generate a specified number of SKs when linear transmission protocols are employed.  The minimum number of transmissions required by a linear protocol to generate $\tau$ secret keys is succinctly given in terms of the inherent $\tau$-connectivity of hypergraph naturally associated with the network.  We have also shown that computing said minimum number of transmissions is \textsf{NP}-hard.

Moreover, we have established that there can be a gap between the number of transmissions required by a nonlinear transmission scheme and the number of transmissions required by the best linear transmission scheme, and that this gap can be arbitrarily large. The problem of characterizing the number of public  transmissions required by a nonlinear  scheme remains an open problem, and appears to be very challenging.

\bibliographystyle{IEEEtran}
\bibliography{sharingUR}

\begin{thebibliography}{10}
\providecommand{\url}[1]{#1}
\csname url@samestyle\endcsname
\providecommand{\newblock}{\relax}
\providecommand{\bibinfo}[2]{#2}
\providecommand{\BIBentrySTDinterwordspacing}{\spaceskip=0pt\relax}
\providecommand{\BIBentryALTinterwordstretchfactor}{4}
\providecommand{\BIBentryALTinterwordspacing}{\spaceskip=\fontdimen2\font plus
\BIBentryALTinterwordstretchfactor\fontdimen3\font minus
  \fontdimen4\font\relax}
\providecommand{\BIBforeignlanguage}[2]{{%
\expandafter\ifx\csname l@#1\endcsname\relax
\typeout{** WARNING: IEEEtran.bst: No hyphenation pattern has been}%
\typeout{** loaded for the language `#1'. Using the pattern for}%
\typeout{** the default language instead.}%
\else
\language=\csname l@#1\endcsname
\fi
#2}}
\providecommand{\BIBdecl}{\relax}
\BIBdecl

\bibitem{bib:PK_CCDE_ISIT2014}
T.~A. Courtade and T.~Halford, ``Coded cooperative data exchange for a secret
  key,'' in \emph{Proc. of IEEE International Symposium on Information Theory
  (ISIT)}, Jun. 2014.

\bibitem{bib:CsiszarNarayanIT2004}
I.~Csisz{\'a}r and P.~Narayan, ``Secrecy capacities for multiple terminals,''
  \emph{IEEE Trans. on Inf. Theory}, vol.~50, no.~12, pp. 3047 -- 3061, Dec.
  2004.

\bibitem{bib:ElRouayheb2010ITW}
S.~El~Rouayheb, A.~Sprintson, and P.~Sadeghi, ``On coding for cooperative data
  exchange,'' in \emph{IEEE Information Theory Workshop (ITW)}, Jan. 2010.

\bibitem{bib:CourtadeWesel2013}
T.~Courtade and R.~Wesel, ``Coded cooperative data exchange in multihop
  networks,'' \emph{IEEE Trans. Inf. Theory}, vol.~60, no.~2, pp. 1136--1158,
  2014.

\bibitem{bib:SprintsonSadeghiISIT2010}
A.~Sprintson, P.~Sadeghi, G.~Booker, and S.~El~Rouayheb, ``A randomized
  algorithm and performance bounds for coded cooperative data exchange,'' in
  \emph{Proc. of IEEE International Symposium on Information Theory (ISIT)},
  June 2010, pp. 1888 --1892.

\bibitem{bib:SprintsonQshine2010}
------, ``Deterministic algorithm for coded cooperative data exchange,'' in
  \emph{ICST QShine}, Nov. 2010.

\bibitem{bib:CourtadeWeselAllerton2010}
T.~Courtade and R.~Wesel, ``Efficient universal recovery in broadcast
  networks,'' in \emph{Proc. of the 48th Annual Allerton Conference on
  Communication, Control, and Computing (Allerton)}, Oct. 2010, pp. 1542
  --1549.

\bibitem{bib:CourtadeWeselAllerton2011}
------, ``Weighted universal recovery, practical secrecy, and an efficient
  algorithm for solving both,'' in \emph{Proc. of the 49th Annual Allerton
  Conference on Communication, Control, and Computing (Allerton)}, Oct. 2011.

\bibitem{bib:GonenLangberg12}
M.~Gonen and M.~Langberg, ``Coded cooperative data exchange problem for general
  topologies,'' in \emph{Proc. of IEEE International Symposium on Information
  Theory (ISIT)}, July 2012, pp. 2606--2610.

\bibitem{bib:Milosavljevic2012}
N.~Milosavljevic, S.~Pawar, S.~E. Rouayheb, M.~Gastpar, and K.~Ramchandran,
  ``Data exchange problem with helpers,'' \emph{arXiv preprint
  arXiv:1202.1612}, 2012.

\bibitem{bib:Milosavljevic}
N.~{Milosavljevic}, S.~{Pawar}, S.~{El Rouayheb}, M.~{Gastpar}, and
  K.~{Ramchandran}, ``Optimal deterministic polynomial-time data exchange for
  omniscience,'' \emph{arXiv preprint:1108.6046 [cs.IT]}, Aug. 2011.

\bibitem{bib:TajbakhshSadeghiShams2011}
S.~Tajbakhsh, P.~Sadeghi, and R.~Shams, ``A generalized model for cost and
  fairness analysis in coded cooperative data exchange,'' in \emph{Proc. Intl.
  Symp. Network Coding (NetCod)}, Beijing, July 2011, pp. 1--6.

\bibitem{bib:OzgulSprintsonITA2011}
D.~Ozgul and A.~Sprintson, ``An algorithm for cooperative data exchange with
  cost criterion,'' in \emph{Information Theory and Applications Workshop
  (ITA), 2011}, Feb. 2011, pp. 1 --4.

\bibitem{bib:HouHsuSprintson13}
I.-H. Hou, Y.-P. Hsu, and A.~Sprintson, ``Truthful and non-monetary mechanism
  for direct data exchange,'' in \emph{Proc. of the 51st Annual Allerton
  Conference on Communication, Control, and Computing (Allerton)}, Monticello,
  IL, Oct. 2013.

\bibitem{chan2011linear}
C.~Chan, ``Linear perfect secret key agreement,”,'' in \emph{2011 IEEE
  Information Theory Workshop Proceedings (ITW2011)}, 2011.

\bibitem{tyagi2013}
H.~Tyagi, ``Common information and secret key capacity,'' \emph{IEEE Trans. on
  Inf. Theory}, vol.~59, no.~9, pp. 5627--5640, 2013.

\bibitem{bib:MukherjeeKashyap2014}
M.~Mukherjee and N.~Kashyap, ``On the communication complexity of secret key
  generation in the multiterminal source model,'' \emph{arXiv preprint
  arXiv:1401.1117}, 2014.

\bibitem{yan2013algorithms}
M.~Yan and A.~Sprintson, ``Algorithms for weakly secure data exchange,'' in
  \emph{Network Coding (NetCod), 2013 International Symposium on}.\hskip 1em
  plus 0.5em minus 0.4em\relax IEEE, 2013, pp. 1--6.

\bibitem{dau2014existence}
S.~H. Dau, W.~Song, and C.~Yuen, ``On the existence of mds codes over small
  fields with constrained generator matrices,'' \emph{arXiv preprint
  arXiv:1401.3807}, 2014.

\bibitem{bib:HalfordCourtadeChugg2013}
T.~Halford, T.~Courtade, and K.~Chugg, ``Energy-efficient, secure group key
  agreement for ad hoc networks,'' in \emph{Communications and Network Security
  (CNS), 2013 IEEE Conference on}, Washington, DC, Oct 2013, pp. 181--188.

\bibitem{effros2012equivalence}
M.~Effros, S.~{El Rouayheb}, and M.~Langberg, ``An equivalence between network
  coding and index coding,'' \emph{arXiv:1211.6660}, Nov. 2012.

\bibitem{alon2008broadcasting}
N.~Alon, E.~Lubetzky, U.~Stav, A.~Weinstein, and A.~Hassidim, ``Broadcasting
  with side information,'' in \emph{Foundations of Computer Science, 2008.
  FOCS'08. IEEE 49th Annual IEEE Symposium on}.\hskip 1em plus 0.5em minus
  0.4em\relax IEEE, 2008, pp. 823--832.

\bibitem{lubetzky2009nonlinear}
E.~Lubetzky and U.~Stav, ``Nonlinear index coding outperforming the linear
  optimum,'' \emph{Information Theory, IEEE Transactions on}, vol.~55, no.~8,
  pp. 3544--3551, 2009.

\bibitem{blasiak2011lexicographic}
A.~Blasiak, R.~Kleinberg, and E.~Lubetzky, ``Lexicographic products and the
  power of non-linear network coding,'' in \emph{Foundations of Computer
  Science (FOCS), 2011 IEEE 52nd Annual Symposium on}.\hskip 1em plus 0.5em
  minus 0.4em\relax IEEE, 2011, pp. 609--618.

\bibitem{frank2003decomposing}
A.~Frank, T.~Kir{\'a}ly, and M.~Kriesell, ``On decomposing a hypergraph into<
  i> k</i> connected sub-hypergraphs,'' \emph{Discrete Applied Mathematics},
  vol. 131, no.~2, pp. 373--383, 2003.

\bibitem{nash1961edge}
C.~S.~J. Nash-Williams, ``Edge-disjoint spanning trees of finite graphs,''
  \emph{Journal of the London Mathematical Society}, vol.~1, no.~1, pp.
  445--450, 1961.

\bibitem{tutte1961problem}
W.~T. Tutte, ``On the problem of decomposing a graph into n connected
  factors,'' \emph{Journal of the London Mathematical Society}, vol.~1, no.~1,
  pp. 221--230, 1961.

\bibitem{bib:Schrijver2003}
A.~Schrijver, \emph{Combinatorial Optimization: Polyhedra and
  Efficiency}.\hskip 1em plus 0.5em minus 0.4em\relax Berlin: Springer-Verlag,
  2003.

\bibitem{bib:NitinawaratNarayan2010}
S.~Nitinawarat and P.~Narayan, ``Perfect omniscience, perfect secrecy, and
  steiner tree packing,'' \emph{Information Theory, IEEE Transactions on},
  vol.~56, no.~12, pp. 6490--6500, Dec 2010.

\end{thebibliography}

\end{document}